\documentclass[10pt,journal]{IEEEtran}

\usepackage[utf8]{inputenc}
\usepackage[T1]{fontenc}
\usepackage{graphicx}
\usepackage{amsmath, amsfonts, amssymb}
\usepackage{amsthm}
\usepackage{mathtools}
\usepackage{xspace}
\usepackage[nocompress]{cite}
\usepackage{cancel}
\usepackage{booktabs}

\usepackage{cleveref}
\usepackage{hyperref}
\crefname{section}{Section}{Sections}
\crefname{algorithm}{Algorithm}{Algorithms}
\crefname{equation}{Equation}{Equations}
\crefname{figure}{Figure}{Figures}
\crefname{table}{Table}{Tables}
\crefname{theorem}{Theorem}{Theorems}
\crefname{lemma}{Lemma}{Lemmas}

\newtheorem{theorem}{Theorem}
\newtheorem{lemma}{Lemma}
\newtheorem{remark}{Remark}

\newcommand{\mcf}{misalignment coefficient\xspace}
\newcommand{\mcfs}{misalignment coefficients\xspace}
\newcommand{\Mcf}{Misalignment coefficient\xspace}

\newcommand{\inv}[1]{\frac{1}{#1}}
\newcommand{\norm}[1]{\left\lVert#1\right\rVert}
\newcommand{\dd}[2][t]{\frac{d}{d #1}#2}

\newcommand{\nvertices}{N}
\newcommand{\ndim}{D}

\newcommand{\adj}[1]{\gamma_{#1}}

\newcommand{\tvel}[2]{\vec{v}_{#1}(#2)}

\newcommand{\tnvel}[2]{\frac{\tvel{#1}{#2}}{\norm{\tvel{#1}{#2}}}}
\newcommand{\coupling}{\alpha}
\newcommand{\vdeg}[1]{k_{#1}}
\newcommand{\adeg}{\langle k \rangle}

\newcommand{\tmis}[2]{H_{#1}(#2)}

\newcommand{\pin}{p_{\text{in}}}

\newcommand{\vit}{\tvel{i}{t}}

\newcommand{\svel}[1]{\vec{x}_{#1}}
\newcommand{\tsvel}[2]{\vec{x}_{#1}(#2)}
\newcommand{\snorm}[1]{z_{#1}}
\newcommand{\tsnorm}[2]{z_{#1}(#2)}
\newcommand{\sit}{\tsvel{i}{t}}
\newcommand{\sjt}{\tsvel{j}{t}}
\newcommand{\sqt}{\tsvel{q}{t}}
\newcommand{\nit}{\tsnorm{i}{t}}

\newcommand{\ki}{\vdeg{i}}
\newcommand{\gij}{\adj{ij}}
\newcommand{\giq}{\adj{iq}}
\newcommand{\si}{\svel{i}}
\newcommand{\sj}{\svel{j}}
\newcommand{\sq}{\svel{q}}
\newcommand{\nni}{\snorm{i}}

\newcommand{\neii}{\svel{\mathcal{N}_i}}

\begin{document}

\title{Network community detection via iterative edge removal in a flocking-like system}

\author{%
  Filipe Alves Neto Verri,
  Roberto Alves Gueleri,
  Qiusheng Zheng,
  Junbao Zhang,
  Liang Zhao
  \IEEEcompsocitemizethanks{%
    \IEEEcompsocthanksitem F.A.N. Verri and R.A. Gueleri are with
    the Institute of Mathematics and Computer Science, University of São Paulo,
    São Carlos, Brazil.
  \IEEEcompsocthanksitem Q. Zheng and J. Zhang are with the
    School of Computer Science, Zhongyuan University of Technology, Zhengzhou,
    China.
  \IEEEcompsocthanksitem L. Zhao is with
    Ribeirão Preto School of Philosophy, Science and Literature, University of São Paulo,
    Ribeirão Preto, Brazil.\protect\\
    E-mail: zhao@usp.br
  }
}


\maketitle


\begin{abstract}
  We present a network community-detection technique based on properties that
  emerge from a nature-inspired system of aligning particles.  Initially, each
  vertex is assigned a random-direction unit vector.  A nonlinear dynamic law is
  established so that neighboring vertices try to become aligned with each other.
  After some time, the system stops and edges that connect the least-aligned
  pairs of vertices are removed.  Then the evolution starts over without the
  removed edges, and after enough number of removal rounds, each community
  becomes a connected component.  The proposed approach is evaluated using
  widely-accepted benchmarks and real-world networks.  Experimental results
  reveal that the method is robust and excels on a wide variety of networks.
  Moreover, for large sparse networks, the edge-removal process runs in
  quasilinear time, which enables application in large-scale networks.
\end{abstract}

\begin{IEEEkeywords}
  Community detection, modularity optimization, flocking formation, complex
  networks.
\end{IEEEkeywords}



\section{Introduction}
\label{sec:introduction}

\IEEEPARstart{T}{he} study of complex networks attracts many researches from
different areas.  Networks are graphs that represent the relationships among
individuals in many real-world complex systems.  Each vertex is an object of
study, and an edge exists if its endpoints interact
somehow~\cite{verri:2016,Pizzuti2017}.

A community structure is commonly found in many real networks, such as social
networks~\cite{Cai2015,Lv2016}, oil-water flow
structure~\cite{gao:2015}, human mobility networks~\cite{thiemann:2010}, spatial
structure of urban movement in large cities~\cite{zhong:2014}, corporate elite
networks in the fields of politics and economy~\cite{heemskerk:2016}, and many
more.  Formally, communities are groups of densely connected
vertices~\cite{clauset:2004,quiles:2016,silva:2016}, while connections between
different communities are sparser.

The problem of community detection is related to the graph
partition problem in graph theory. Finding the optimal partition is
an NP-hard problem in most cases~\cite{fortunato:2010}, thus making room for
many researches to find out sub-optimal solutions in feasible time.
As a result, various approaches for community detection have been developed, including
spectral properties of graph matrices~\cite{mitrovic:2009,Agliari2017},
particle walking and competition in networks~\cite{silva:2016, verri:2016}, and
many evolutionary or bio-inspired processes~\cite{Pizzuti2017}.


Newman and Girvan~\cite{newman:2004} have proposed a metric called
\emph{modularity}, whose purpose is to quantify how a network is likely to
display community structure~\cite{lancichinetti:2011}.  It does not make any
assumption on \textit{a-priori} knowledge of the network, e.g., vertex labels.
One of the algorithms employed in this paper, both for comparison and for
complementary stage, is called ``Cluster Fast Greed'', or just CFG, and it is
based on a greedy optimization of modularity~\cite{clauset:2004,
lancichinetti:2009}. It performs really fast, in time $O(md \, \text{log} \,
n)$, where $d$ is the depth of the dendrogram describing the network's
hierarchical community structure returned by the algorithm. In cases where $m
\sim n$ and $d \sim \text{log} \, n$, it runs in quasilinear time: $O(n \,
\text{log}^2 n)$. Another algorithm employed here is the so-called ``Louvain'',
which is based on the optimization of modularity too~\cite{blondel:2008,
lancichinetti:2009}. The authors advocate in favour of the computation time,
which makes the algorithm applicable on huge networks.

In this paper, we propose a bio-inspired community-detection method that is
divided in two alternating stages. The first one is a nonlinear collective complex system
that takes inspiration from the flocking formation in nature~\cite{nagai:2015}.
In the second stage, we measure the misalignment of each pair of vertices that
are directly connected and remove a fraction of edges that result the
highest misalignments.
Flocks are groups of individuals that move in a coordinated fashion.  This
coordinated motion emerges even in the absence of any leader, what makes it a
self-organizing phenomenon.  In our model, each vertex is an aligning
particle. Therefore, each vertex carries a
velocity vector, pointing to a random direction at the beginning and, as the
process evolves, progressively turns itself toward the same direction of its
neighboring vertices.  As a simplification of the process, the vertices
actually do not move, so the term
``velocity'' is just an analogy to the direction of motion in flocking systems.
The dynamical process is suspended after a certain number of iterations, then
the second stage takes place.  The edge that connects the least aligned pair of
vertices is supposed to link distinct communities.
After enough number of removal cycles, most of the inter-community edges are
expected to be removed, thus the network becomes partitioned into disconnected
components.

Our model is evaluated using widely-accepted benchmarks and real-world networks.
It not only excels on many different scenarios but also has very low
computational cost.  As a result, the research shows potential for a broad range
of applications, including big data.

The rest of this paper is organized as follows.  \Cref{sec:model,sec:theory} describe
the proposed model and present analytical results of the system.
In \cref{sec:experimental}, computer
simulations illustrate the process and assess its performance.  Finally,
\cref{sec:conclusions} discusses and concludes this paper.


\section{Model description}
\label{sec:model}

Here we present the iterative edge-removal approach in detail. As mentioned
before, the process is divided in two alternating stages: the particle-alignment
dynamical process and the edge-removal stage itself.

\subsection{The particle-alignment dynamical model}

Given an undirected network, free of self-loops and multiple edges, we let each
vertex $i \in \{1, 2, \dots, \nvertices\}$ be a particle in the collective
dynamical system.  Therefore, each vertex carries a velocity vector
$\mathbf{v}_i(t) \in \mathbb{R}^\ndim$, which points to some direction in a
multi-dimensional space.  The vertices actually do not move at all, so the term
``velocity'' is just an analogy to the direction of motion in flocking systems,
whose moving particles try to align themselves to their neighbors, i.e., try to
take the same direction of motion.

Edges of the network define the neighborhood of each particle $i$, i.e., those
particles to which $i$ can directly interact.  We denote $\adj{ij} = 1$ if
vertices $i$ and $j$ interact, and $\adj{ij} = 0$ otherwise.  The neighborhood
is unchangeable throughout the dynamical process, thus the interaction network
is the same all the time.

Initially, each particle $i$ is assigned a random initial velocity
$\mathbf{v}_i(0)$, which is a unit vector pointing to a random direction. A
three-dimensional space ($\ndim= 3$) is employed for all the experiments
presented in this paper. However, how dimensionality impact on the final
results has to be investigated.

The nonlinear dynamical system is governed by the expression
\begin{multline}
  \label{eq:v}
  \tvel{i}{t+1} = \tnvel{i}{t} + \\
  \coupling \inv{\vdeg{i}} \sum_{j=1}^{\nvertices}{
    \adj{ij}\left(\tnvel{j}{t} - \tnvel{i}{j}\right)
  }\text{,}
\end{multline}
where $\norm{\cdot}$ denotes the Euclidean norm.  The nonlinearity of the
dynamical system is introduced by the velocity normalization. Parameter $t$ is
the iteration index (time), which starts from zero. Parameter $\coupling > 0$
defines how fast the velocities are updated.  The value $\vdeg{i}$ is the degree
of vertex $i$, i.e., the number of neighbors it has,
\begin{equation}
  \vdeg{i} = \sum_{j=1}^{\nvertices}{\adj{ij}}\text{.}
\end{equation}


Although the particles do not move --  there is even no position defined to
them --, by using the unit vector we enforce something
analogous to the ``constant-speed motion'' that each particle would perform. In addition,
it enforces that at any given time-step $t$, all the particles would move at the same
speed. Constant-speed motion is a property commonly modelled in studies of self-propelled
particles~\cite{vicsek:2012, nagai:2015}.
It is also responsible for the rotational symmetry breaking that makes the set of
particles agree on the same velocity direction.  By not enforcing such normalization, all
velocities would vanish to zero --  or close to zero  -- due to their random initial
distribution, making opposite vectors neutralize each other.
We show in \cref{sec:theory} that particles in the same
connected component are most likely to align.  As a result, velocity vectors of particles from
different communities will converge to the same value.  However, as presented in
\cref{sub:illustration}, particles in the same community tend to align from
different direction, which motivates the removal of edges some time before the
convergence.

\subsection{The iterative edge-removal process}

We define the \emph{\mcf} $\tmis{ij}{t}$ as the level of disorder in terms of
velocity-vectors' misalignment between nodes $i$ and $j$.  Such index is
mathematically expressed by
\begin{equation}
  \tmis{ij}{t} = d_1\!\left(\tnvel{i}{t}, \tnvel{j}{t}\right)
\end{equation}
where $d_1(\cdot,\cdot)$ is the $L_1$ distance between vectors. In other words,
$\tmis{ij}{t}$ is just the city-block distance between the normalized velocity
vectors of the vertices $i$ and $j$.

As we will see through the experiments presented in the next section,
\mcfs of edges that connect distinct communities decrease slower than
those connecting vertices inside the same community. It builds the basis of our
iterative edge-removal approach: removing the edges with highest \mcfs is
likely to remove inter-community edges, thus making the distinction between
different communities clearer and clearer over removal cycles. So the overall
edge-removal process consists of the following steps:
\begin{enumerate}
  \item After assigning random initial velocity vectors to every vertex, run the
    dynamical particle-aligning model, as described in the previous section, up
    to some number of time-steps. (Number of steps is discussed in \cref{sub:misalign}.)
  \item Once the dynamics is interrupted, collect the misalignment coefficient
    of each pair of interacting vertices. It
    is possible to run the dynamical model (Step~1) multiple times, using
    different random assignments to the set of initial velocities. In this case,
    the final coefficient of each edge can be just the summation of individual
    coefficients collected after each run.
  \item Remove the edges with the highest \mcfs, then go to Step~1 and run the
    dynamical model again, but this time using the new network without the
    removed edges.
\end{enumerate}
Steps~1, 2, and 3, together, form what we call ``cycle'' or ``round''. In this
paper, we study the influence of running the dynamical model multiple times per
cycle. We also study the influence of removing different numbers of edges per
cycle.


\section{Theoretical results on flocking alignment}
\label{sec:theory}

We also present analytical and argumentative study of the dynamics given
by \cref{eq:v}.

\subsection{Domain of the velocity vectors}

A state at time $t + 1$ in which $\norm{\vit} = 0$ for any $i$ is a singularity.
In order to deal with this problem, we need to restrict the parameter $\alpha$.
Once $\alpha > 0$ by definition, we show that, if $\alpha < 0.5$ and
$\norm{\tvel{i}{0}} = 1$ for all $i$, then $0 < \norm{\vit} \leq 1$ for all $i$
and $t > 0$.  Thus, for any reasonable $\alpha$ and the initial conditions
proposed in this paper, the velocity vectors are nonzero vector with norm less
than or equal to $1$.
%

Given the restrictions of the parameters and initial conditions, the following
lemmas prove that $\norm\vit \in (0, 1]$ for all $i$ at any time $t$,
guaranteeing the expected behavior of the model.

\begin{lemma}
  Given that $\norm{\tvel{i}{0}} \leq 1$, for all particle $i$, the norm of the velocity
  of every particle will not surpass $1$.
\end{lemma}

\begin{proof}
  For $t = 0$, $\norm{\tvel{i}{0}} \leq 1$ from the statement.

  Assume $\norm{\tvel{i}{t}} \leq 1$ for some $t$. Then,
  \begin{multline}
    \norm{\tvel{i}{t+1}} = \norm{
      \left(1 - \alpha\right) \tnvel{i}{t} +
      \alpha \inv{\vdeg{i}} \sum_{j}{ \adj{ij} \tnvel{j}{t} }
    } \leq {} \\
    \left(1 - \alpha\right) \norm{ \tnvel{i}{t} } +
    \alpha \inv{\vdeg{i}} \sum_{j}{ \adj{ij} \norm{ \tnvel{j}{t} } } =\\
    \left( 1 - \alpha \right) + \alpha = 1 \implies
    \norm{\tvel{i}{t+1}} \leq 1
  \end{multline}

  Thus, by induction, $\norm{\tvel{i}{t}} \leq 1$ for all $i$, $t$.
\end{proof}

\begin{lemma}
  Given that $\norm{\tvel{i}{0}} > 0$, for all particle $i$, the norm of the velocity of
  every particle is strictly greater than $0$ for any time $t>0$.
\end{lemma}

\begin{proof}
  For $t = 0$, $\norm{\tvel{i}{0}} > 0$ from the statement.

  Assume that $\norm{\tvel{i}{t}} > 0$ for some $t$.
  We show by contradiction that $\norm{\tvel{i}{t + 1}} > 0$.

  If there exists $\alpha=\alpha_0$, $0 < \alpha_0 < \inv{2}$, such that
  $\norm{\tvel{i}{t + 1}} = 0$,
  then
  \begin{multline}
    \vec{0} =
    \left(1 - \alpha_0\right) \tnvel{i}{t} +
    \alpha_0 \inv{\vdeg{i}} \sum_{j}{ \adj{ij} \tnvel{j}{t} } \implies {} \\
    - \tnvel{i}{t} = \alpha_0 \left(
      -\tnvel{i}{t} + \inv{\vdeg{i}} \sum_{j}{ \adj{ij} \tnvel{j}{t} }
    \right) \implies {} \\
    \norm{ -\tnvel{i}{t} } = \alpha_0 \norm{
      -\tnvel{i}{t} + \inv{\vdeg{i}} \sum_{j}{ \adj{ij} \tnvel{j}{t} }
    } \implies {} \\
    1 \leq \alpha_0 \left(
      \norm{ -\tnvel{i}{t} } + \inv{\vdeg{i}} \sum_{j}{ \adj{ij} \norm{\tnvel{j}{t}} }
    \right) =\\
    2 \alpha_0 < 1 \implies {} 1 < 1\text{.}
  \end{multline}

  By contradiction, such $\alpha_0$ does not exist.

  Thus, by induction, $\norm{\tvel{i}{t}} > 0$ for all $i$, $t$.
\end{proof}

\subsection{Alignment of the velocity vectors}

The core mechanism in our method is measuring small misalignments between
connected particles and deciding which edge will be removed.  A
question that rises is whether the velocity vectors converge to a single point
or not, that is, if the particles align or not.  We show that perfect alignment
of particles in the same connected component is most likely to happen.
The system would also be in equilibrium if vectors are in perfect opposition to
each other.  Such condition would need not only very specific initial velocity
vectors but also specific network configuration, thus this case is extremely
unlikely.

To perform the particle-alignment study, we use a continuous approximation of the
evolution equations,
\begin{multline}
  \dd{\vit} =
    \frac{1 - \norm{\vit}}{\norm{\vit}} \vit +\\
    \alpha \inv{\vdeg{i}} \sum_{j}{
      \adj{ij} \left( \tnvel{j}{t} - \tnvel{i}{t} \right)
    }\text{.}
\end{multline}

We are not interested in the evolution of the unnormalized velocities but in their
normalized forms.  To improve readability, we set
\begin{equation}
  \tsvel{i}{t} = \tnvel{i}{t}\text{ and }
\end{equation}
\begin{equation}
  \tsnorm{i}{t} = \norm{\vit}\text{.}
\end{equation}

Thus, the governing equations of the normalized velocities $\tsvel{i}{t}$ are
\begin{equation}
  \begin{dcases}
    \dd{\sit} &= \alpha\inv{\ki\nit}\sum_q{\giq\Big(\sqt - \big(\sit\cdot\sqt\big)\sit\Big)}\\
    \dd{\nit} &= 1 - \nit - \alpha\sum_q{\giq\Big(1 - \sit\cdot\sqt\Big)}
  \end{dcases}
\end{equation}
where $\cdot$ stands for the dot product operator.

\begin{theorem}
  The aligned state $\sit = \sjt$, for all particle $i$ that interacts with particle $j$,
  is stable in the sense of Lyapunov.
\end{theorem}

\begin{proof}
  We define the energy function $E$ that reaches zero only when the aligned state is
  reached, and it increases as the velocities misalign,
  \begin{equation}
    E = \inv{4}\sum_i{\sum_j{\gij\big(\sj - \si\big)\cdot\big(\sj - \si\big)}}\text{.}
  \end{equation}

  Its derivative is
  \begin{equation}
    \dd{E} = \inv{2}\sum_i\sum_j{\gij\big(\sj - \si\big)\cdot\left(\dd\sj -
    \dd\si\right)}\text{,}
  \end{equation}
  but
  \begin{equation}
    \si\cdot\dd\si =
    \alpha\inv{\ki\nni}\sum_q{
      \giq\left(\si\cdot\sq - \si\cdot\sq\right)
    } = 0\text{,}
  \end{equation}
  then
  \begin{multline}
    \dd{E} = -\sum_i\sum_j{\gij \sj\cdot\dd\si} = \\
    -\sum_i\sum_j\sum_q{
      \alpha\gij\giq\inv{\ki\nni}\Big(\sj\cdot\sq - \big(\si\cdot\sq\big)\big(\si\cdot\sj\big)\Big)
    } = {}\\
    \sum_i\sum_j\sum_q{
      \alpha\gij\giq\inv{\ki\nni}\big(\si\cdot\sq\big)\big(\si\cdot\sj\big)
    } - {}\\
    \sum_i\sum_j\sum_q{
      \alpha\gij\giq\inv{\ki\nni}\sj\cdot\sq
    } = {}\\
    \alpha\sum_i\inv{\ki\nni}\sum_j{
      \gij\big(\si\cdot\sj\big)\sum_q\giq\big(\si\cdot\sq\big)
    } - {}\\
    \alpha\sum_i\inv{\ki\nni}\sum_j{
      \gij\sum_q\giq\big(\sj\cdot\sq\big)
    } = {}\\
    \alpha\sum_i\inv{\ki\nni}\sum_j{
      \gij\big(\si\cdot\sj\big)\si\cdot\sum_q\giq\sq
    } - {}\\
    \alpha\sum_i\inv{\ki\nni}\sum_j{
      \gij\sj\cdot\sum_q\giq\sq
    }\text{.}
  \end{multline}

  Let $\neii = \sum_q\giq\sq$ be the summation over all the velocities of the neighbors of
  particle $i$, then
  \begin{multline}
    \dd{E} = \alpha\sum_i\inv{\ki\nni}\sum_j{
      \gij\big(\si\cdot\sj\big)\big(\si\cdot\neii\big)
    } -\\
    \alpha\sum_i\inv{\ki\nni}\sum_j{
      \gij\big(\sj\cdot\neii\big)
    } = {}\\
    \alpha\sum_i\inv{\ki\nni}\big(\si\cdot\neii\big)\si\cdot\sum_j{
      \gij\sj
    } -
    \alpha\sum_i\inv{\ki\nni}\neii\cdot\sum_j{
      \gij\sj
    } = {}\\
    \alpha\sum_i\inv{\ki\nni}\big(\si\cdot\neii\big)\big(\si\cdot\neii\big) -
    \alpha\sum_i\inv{\ki\nni}\big(\neii\cdot\neii\big)\text{.}
  \end{multline}

  But, given any vectors $\vec{a}, \vec{c} \in \mathcal{R}^D$, we have
  $
    \vec{a}\cdot\vec{c} \leq \norm{\vec{a}}\norm{\vec{c}}
  $
  and
  $
    \vec{a}\cdot\vec{a} = \norm{\vec{a}}^2
  $,
  \begin{multline}
    \dd{E} =
    \alpha\sum_i\inv{\ki\nni}\big(\si\cdot\neii\big)^2 -
    \alpha\sum_i\inv{\ki\nni}\big(\neii\cdot\neii\big) \leq {}\\
    \alpha\sum_i\inv{\ki\nni}\big(\cancelto{1}{\norm{\si}}\norm{\neii}\big)^2 -
    \alpha\sum_i\inv{\ki\nni}\norm{\neii}^2 \leq 0
  \end{multline}
\end{proof}

\begin{remark}
  Perfect alignment most likely happens since $\dd{E}$ is zero only if $x_i$ and
  $\neii$ are codirectional, that is, when all neighbors are either aligned or
  opposed to each other.  While this condition does not hold, $\dd{E} < 0$, and
  every particle keeps trying to align with its neighbors.
\end{remark}


\section{Experimental results}
\label{sec:experimental}

In this section, we present an extensive set of experimental results conducted on
different classes of computer-generated and real-world networks.

\subsection{Illustrative example} 
\label{sub:illustration}

In this subsection, we present illustrative examples of the application of
our method.
Input networks are built by employing the following methodology:
\begin{enumerate}
  \item Given a desired average vertex degree $\langle k
    \rangle_{\mathrm{des}}$, each vertex $i$ randomly chooses $\langle k
    \rangle_{\mathrm{des}} / 2$ vertices $j$ to connect to. For each $j$ to be
    selected, $j$ is taken from the same community of $i$ with a probability
    $\pin$, or accordingly taken from a different community with a
    probability $p_{\text{out}} = 1 - \pin$. The selection of the same
    vertex $j$ twice or more is allowed, as it is also allowed for $i$ to
    connect to itself. Also, $j$ being selected by $i$ does not prevent $i$
    being also selected by $j$.
  \item After establishing all the connections, the network is simplified, in
    the sense that loops (self-connections) are removed and multiple edges that
    connect the same pair of vertices become a single, undirected edge. As a
    result, the actual average degree $\langle k \rangle$ may be reduced, but
    for large networks, it remains very close to the desired degree $\langle k
    \rangle_{\mathrm{des}}$.
\end{enumerate}

\begin{figure}
  \centering
  \includegraphics[width=\columnwidth]{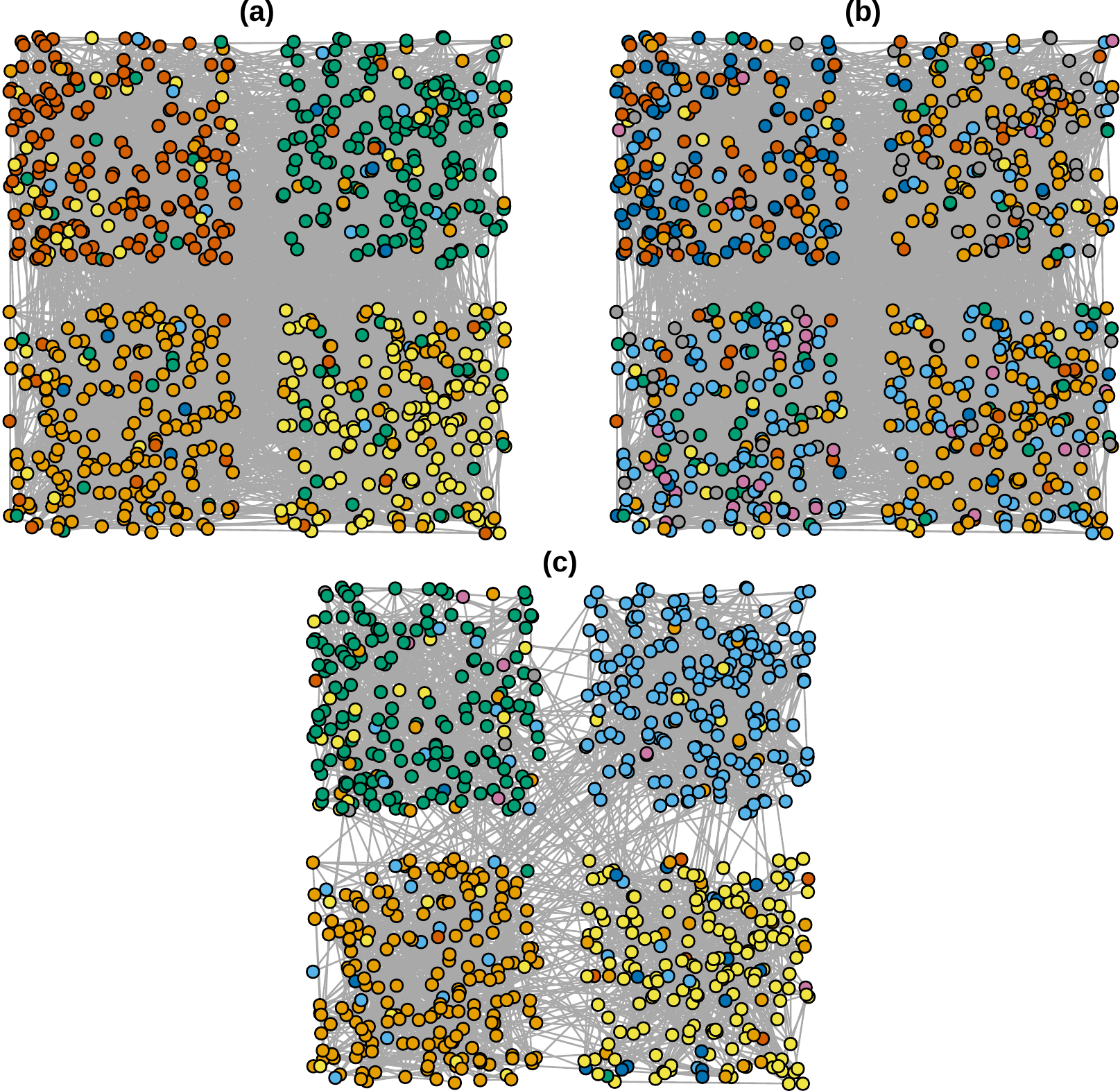}
  \caption{
    Overview of the community-detection results.  The input network has 4
    communities, each one made of 200 vertices, $\langle k \rangle \approx 10$,
    and $\pin = 0.66$.  The results of applying CFG (a) and Louvain (b)
    algorithms on the original network are shown. Communities are represented by
    different colors.  In the last figure (c), we display the results after
    applying 1766 rounds of the proposed iterative edge-removal process. In this
    case, communities are just different connected components. In each cycle, 10
    independent runs are performed and the most misaligned edge is removed.
    Only the remaining edges are shown in (c).  We set $\alpha = 0.1$.
  }
  \label{fig:overview}
\end{figure}

For comparison purpose, in \cref{fig:overview}, we also show the results of
applying CFG and Louvain algorithms on a network with 4 communities, average
degree $\langle k \rangle \approx 10$ and $\pin = 0.66$.
In this case, CFG
achieves modularity $Q = 0.33$ and Louvain, $Q = 0.31$.
We apply our method in this network, setting $\alpha = 0.1$.  In each round, we
run the dynamical system until $t=100$ with $10$ independent initial
configurations.  The most misaligned edge per round is removed until there is
no more edges to remove.
We choose the partition that yields the best modularity.
The results of the iterative edge-removal process are better than those of CFG
and Louvain, reaching $Q = 0.38$.
%

\begin{figure}
  \centering
  \includegraphics{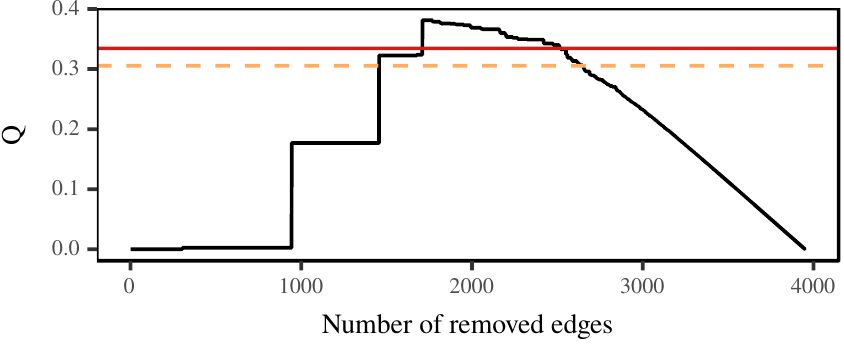}
  \caption{
    Evolution of the modularity $Q$ over the rounds of the proposed iteration
    edge-removal process.  The input network and learning configuration are the
    same presented in \cref{fig:overview}.  Solid and dashed horizontal lines
    correspond to the modularity achieved by CFG and Louvain, respectively.
  }
  \label{fig:modularity}
\end{figure}

\Cref{fig:modularity} shows the modularity evolution over different rounds.
Since communities are defined by connected components, the first edge-removal
rounds just result in a single community, making the modularity score be very
low. Those abrupt transitions reveal different components becoming completely
disconnected from each other, i.e., they reveal those rounds in which the last
edge that links two large components is removed. After achieving four major
communities --- and consequently the highest score ---, further removal starts
destroying them. Such a local optimum partition is the one that should be
returned by the proposed iterative edge-removal method.


\begin{figure}
  \centering
  \includegraphics{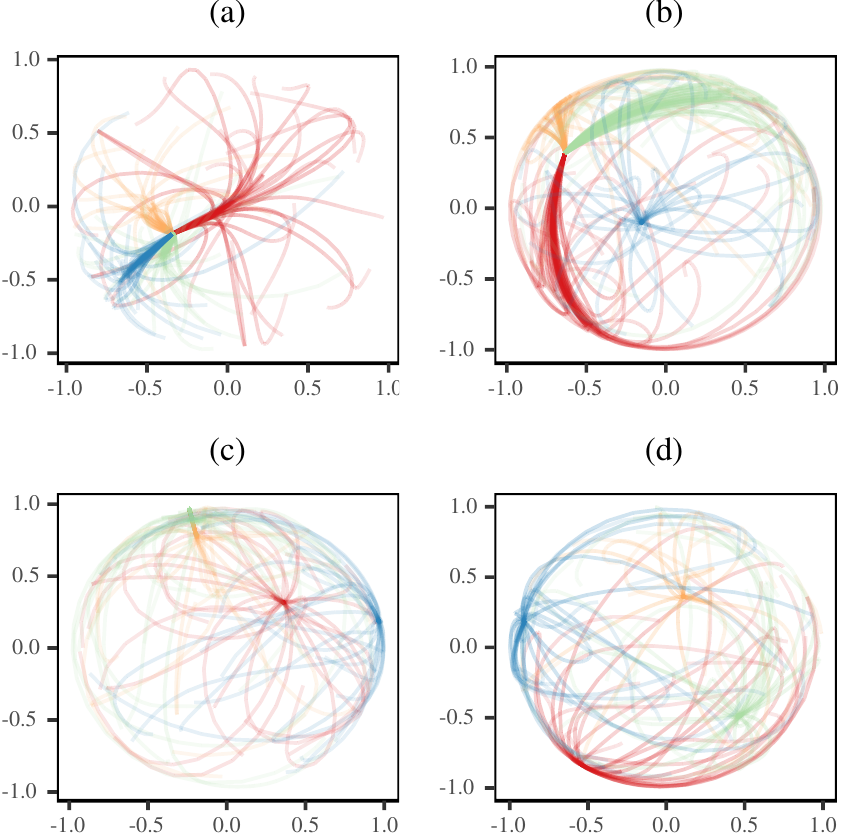}
  \caption{
      Evolution of the dynamics governed by Equation \ref{eq:v} in a network with
      $800$ nodes, average degree $\langle k\rangle = 10$, $4$ communities, and
      $\pin = 0.9$.  For better visualization, even though the system has
      been simulated with $D = 3$ dimensions and $800$ particles, only $200$
      random-selected normalized velocity vectors are shown projected in two
      dimensions.  $1000$ system iterations are run in each plot and the path of the
      velocity vectors are more opaque in the last iterations.  Path colors match
      the four different communities.  In plot (a), no edge was removed.
      In the remaining plots, 200 (b), 400 (c) and 600 (d) edges have been removed
      iteratively.
  }
  \label{fig:dynamics}
\end{figure}

To further illustrate our method, we provide another simulation.  We show (in
\cref{sec:theory}) that all the vertices in a connected component will align as time tends to
infinity.  As a result, one might wonder how our edge-removal strategy works in
practice.  To demonstrate its effectiveness, the method is applied on a simpler
community detection problem: detecting the $4$ communities with the same size in a
random clustered network that comprises $800$ nodes with average degree $\langle
k \rangle = 10$.  Connections are randomly assigned such that every node has
probability $\pin = 0.9$ of connecting to another node in its community.  We run
our method with $3$-dimensional velocity vectors and $\alpha = 0.05$.
For the sake of visualization, \cref{fig:dynamics} depicts the evolution
of the normalized velocity vectors (projected in two dimensions) in different
situations.  Line colors represent communities, and the transparency decreases
in function of time $t$.
Plot (a) depicts $1000$ iterations of our dynamical system in the original
network.  One can notice that all particles are aligned, but the velocity
vectors of particles in different communities converge from different
directions.  This phenomenon is utterly important, since it enables us to
distinguish the communities.  It also explains the difference between our method
and Kuramoto-based ones\cite{Oh2005,Arenas2008}.  In the
synchronization-based techniques, each element usually is a fixed low
dimensional dynamical system.  In the Kuramoto oscillator model, only a single
real value is associated to each node, which corresponds to the phase.  Thus,
nodes can only synchronize ``from two different directions''.  Using three or
more dimensions in our method, we bring an infinitude of possible directions.
In the same figure, we also show three snapshots of the edge-removal process.
After running up to $t=1000$, we remove the $10$
edges with highest values of \mcf.  We repeat this process until $600$ edges
have been removed.  The evolution of the normalized velocity vectors at $t = 1, 2,
\dots, 1000$ after removing $200$, $400$, and $600$ edges are illustrated in
subplots (a), (b), and (c), respectively.
After the removal of $200$ edges, we observe that one of the communities
disconnects from the others, becoming a single connected component, and thus,
the velocity vectors converge to a different point.  With $400$ edges removed,
another community detaches.  And finally, after $600$ removals, each
community becomes a connected component of the network, achieving our goal.

\subsection{Analysis of the evolution of misalignment coefficient}
\label{sub:misalign}

\begin{figure}
  \includegraphics{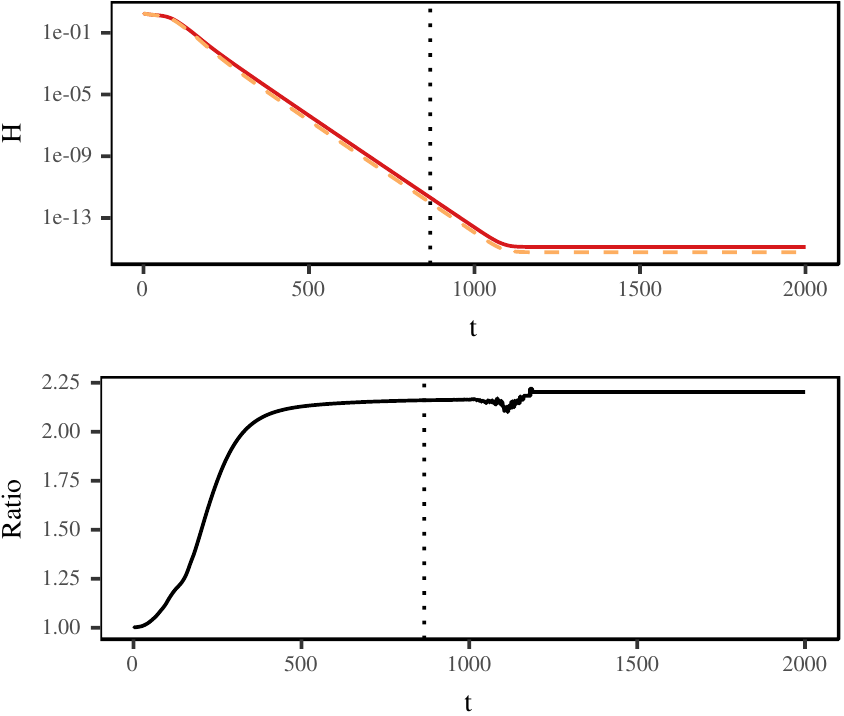}
  \caption{
    \Mcf of the network presented in \cref{fig:overview}: 4
    communities, each one made of 200 vertices, $\langle k \rangle \approx 10$,
    and $\pin = 0.66$. The population of coefficients is obtained from 10
    independent runs. Parameter $\alpha = 0.1$.  Measurements are taken from a
    single round.  Top plot: average \mcf evolution separated between intra-
    (dashed line) and inter-community (solid line) edges.
    Bottom plot: ratio between the \mcf of inter- and intra-community edges.
  }
  \label{fig:entropy}
\end{figure}

In the previous section, we claimed that the \mcfs of edges connecting distinct
communities become usually higher than those connecting vertices inside the same
community. Let us now present, in \cref{fig:entropy}, how
\mcfs change over time.

In order to reduce the dependence on the initial condition -- random assignment
of velocities --, the population of \mcf values is obtained from 10 independent
runs.  The input network is the same for all of these runs.  The network has 800
vertices and $\langle k \rangle \approx 10$, what gives a total of approximately
4000 edges.

In the top plot of \cref{fig:entropy}, we plot the average \mcf $H$ grouped by
intra- and inter-edges.  As we can see, coefficients fall down quickly.
Intra-edges, however, align
faster than the inter-edges.  In the bottom plot, we show the ratio between
intra- and inter-edges.  As expected, a greater proportion of intra-edges have
lower misalignment coefficient.  Moreover, after many iterations, the velocity
vectors become almost identical for every pair of connected vertices.  At this
point, the finest machine representation of real numbers is reached.  The
vertical dotted line indicates the iteration $t$ in which the average
misalignment is lower than $10^{-12}$.  Any result beyond this point might be
meaningless.  Consequently, we should always stop the system earlier.

\subsection{Analysis of the number of removed edges}

Results of removing a different number of edges per round are presented in this section.
An evaluation index is used in order to objectively quantify the accuracy of the
set of obtained communities. Specifically, we selected the adjusted Rand index
(ARI), which is the
corrected-for-chance version of the Rand index~\cite{hubert:1985}.  It measures the
similarity between the partition obtained from some algorithm and a reference
partition. ARI generates values between $-1$ and $1$. If two partitions match
perfectly each other, it results in a value $1$. On the other hand, it ensures a
value close to $0$ for a randomly-labelled partition or a partition that assigns
all the elements into a single group, given that the reference partition has
more than one group, of course. Negative values stand for some anti-correlation
between the pair of partitions.

Real-world applications, however, do not provide such a reference partition, so
we also employ the \emph{modularity} score.  Unlike ARI, the modularity score does
not make any assumption on \textit{a-priori} knowledge of the network, e.g.,
vertex labels. By placing the ARI score against the modularity score, we can
check for any relationship between them, i.e., check for a relationship between
an information available in a real application (modularity) and a robust measure
based on reference partitions (ARI).

\begin{figure}
  \centering
  \includegraphics[width=\columnwidth]{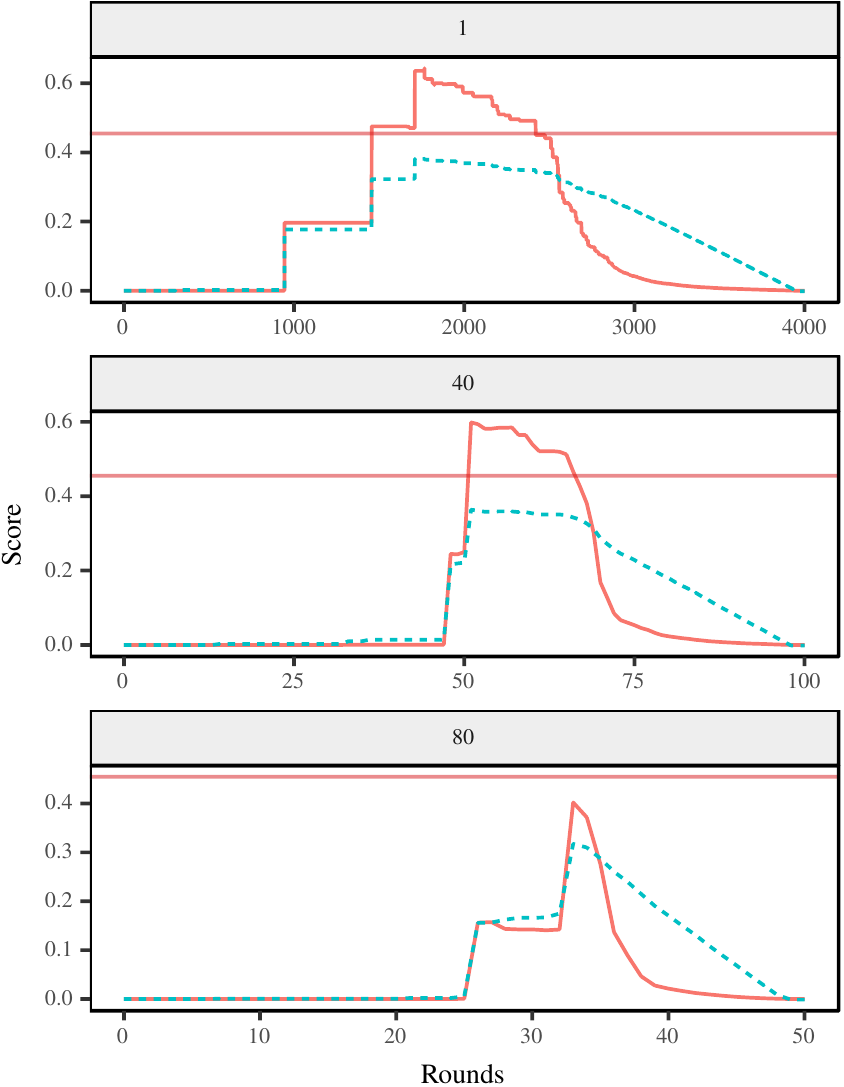}
  \caption{
    ARI (solid lines) and modularity (dashed line) scores along different
    rounds.  The input network is exactly the same one depicted in
    \cref{fig:overview}: 4 communities, each one made of 200 vertices, $\langle
    k \rangle \approx 10$, and $\pin = 0.66$.  Plots are arranged according to
    the amount of edges removed per round: 1, 20, and 80.  For all runs,
    we set $\alpha = 0.1$ and the system runs until $t = 100$.
    The solid horizontal line indicates the ARI score obtained by the CFG
    method.
  }
  \label{fig:nedges}
\end{figure}

Firstly, let us consider the same network presented in \cref{fig:overview}, with
$\pin = 0.66$.  ARI and modularity scores along different rounds are presented in
\cref{fig:nedges}, arranged according to the number of edges removed per
round.

Even when 40 edges ($\approx 1\%$ of the total number of edges) are removed per
round, our method achieves higher ARI score than the CFG method: $0.64$ (1
edge) and $0.60$ (40 edges) against $0.45$.
Removing fewer edges per round is slightly
better.  Removing a larger fraction of the edges is less expensive though, for it
demands fewer rounds to complete.  However, we experience a significantly drop
in the score if the fraction of removed edges per round is too big (around 2\%.)



Moreover, a desirable relationship is noticeable here: the highest
modularity matches relatively high ARI scores.  This
result is useful for establishing when we should stop the edge-removal process
and where the optimal round is, i.e., the round after which we are likely to
find a good partition: stopping the process just before the modularity starts
decreasing and taking such highest-modularity's network, using the connected
components as the final set of communities. Although the results of
\cref{fig:nedges} come from just one network instance, the same overall pattern
is still present in other instances, even using different network parameters, as
we are going to see in the remaining results presented in this paper.

\begin{figure}
  \centering
  \includegraphics{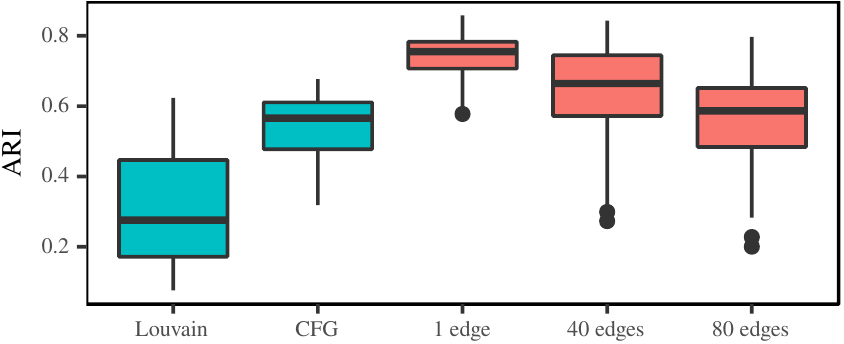}
  \caption{
    Results of removing 1, 40 and 80 edges per round on a population of 50
    different networks. Each network has 4 communities, each one made of 200
    vertices, $\langle k \rangle \approx 10$, and $\pin = 0.66$. Boxes in blue
    are the results of applying CFG and Louvain on the network.  The remaining
    boxes (in red) are obtained by applying the proposed technique removing
    different amounts of edges.  The partition that yields the best modularity
    is chosen.  10 runs per round are performed, and parameter $\alpha = 0.1$
    for all simulations.  We stop the system at $t = 100.$
  }
  \label{fig:nedges_sym}
\end{figure}

In \cref{fig:nedges_sym}, we present a comparison between removing 1, 40, and 80
edges per round on a population of 50 different networks.
In fact, removing a smaller fraction of the edges yields better results.  In
addition, the variability of results is reduced when removing less edges per
round. However, the proposed edge-removal process leads to
considerably better results compared to the application of CFG or Louvain on the
original network.

\subsection{Computational complexity}
\label{sub:complexity}

The edge-removal approach presented here is somewhat similar to the
edge-betweenness-based algorithm proposed by Newman and
Girvan~\cite{newman:2004}. The \mcf of an edge measures how different
the incident vertices are in terms of their velocity vector. Edge
betweenness, in turn, measures the relevance of an edge in terms of the number
of shortest paths that include such an edge. Both concepts ultimately try to
identify edges that connect distinct communities. An important drawback of the
edge-betweenness approach is, however, its cubic time complexity. Precisely, for
a network of $m$ edges and $n$ vertices, the time complexity is $O(m^2n)$, or
$O(n^3)$ for sparse networks, in which case $m \sim n$~\cite{clauset:2004,
newman:2004, lancichinetti:2009, fortunato:2010}.  Therefore, it is necessary to
have a discussion about the complexity of the approach proposed
here.

Four variables are relevant to our time-complexity analysis:
the number $n$ of vertices;
the number $m$ of edges;
the total amount of time-steps $t_\text{total}$ per run;
and the number of edge-removal rounds $n_\text{rounds}$ required to obtain an
appropriate partition of the network.
The dimensionality of the velocity vectors is constant, three dimensions are
employed in our experiments, so it is not relevant to this analysis. Our
hypothesis is that both $t_\text{total}$ and $n_\text{rounds}$ can be invariant
no matter how big the network is.

In the case of removing just a single edge per round, $n_\text{rounds}$ gets the
order of $m$. However, the experiments presented in the previous section
indicate that it is still possible to obtain satisfactory results by removing a
fraction of all the edges per round, in which case $n_\text{rounds}$ becomes
constant. For instance, removing approximately 1\% of the edges per round
requires around 50 rounds to achieve good results. Also, concerning the
number of runs per round, if one decides to perform a set of runs per round in order to reduce
the influence of the initial random configuration, the number is still
relatively small and does not scale with neither $n$ nor $m$.

In order to provide evidence that supports our hypothesis --  neither
$t_\text{total}$ nor $n_\text{rounds}$ scales with the network size --, we
present some results on bigger and denser networks.
Particularly, a network that has 10 times more vertices than those
studied in the previous section and a network that has twice more edges while
keeping the same number of vertices.

\begin{figure}
  \centering
  \includegraphics[width=\columnwidth]{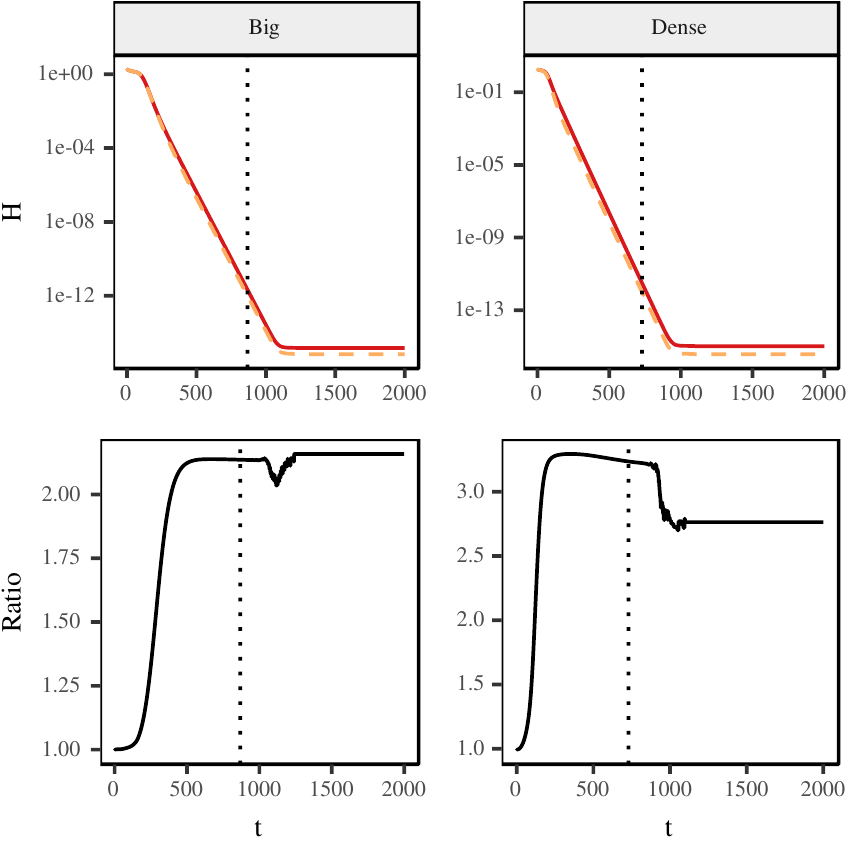}
  \caption{
    \Mcf of a big and a dense network: both of them with 4 communities and $\pin
    = 0.66$.  The big network has 2000 vertices in each community, while the
    dense has 200.  The average degree in the big network is $\adeg \approx 10$,
    while in the dense one,  $\adeg \approx 20$.  The population of coefficients
    is obtained from 10 independent runs. Parameter $\alpha = 0.1$.
    Measurements are taken from a single round.  Top plots: average \mcf
    evolution separated between intra- (dashed line) and inter-community (solid
    line) edges.  Bottom plots: ratio between the \mcf of inter- and
    intra-community edges.
  }
  \label{fig:bignet}
\end{figure}

In \cref{fig:bignet}, we present the evolution of the coefficients in the first
round.  We notice that increasing the network size or the density has small
effects on the amount of
time-steps needed to obtain satisfactory separation.
Interestingly, the bigger network requires a little more iterations while the
denser one, a little less.
Therefore we can estimate that, for an arbitrarily large network,
$t_\text{total}$ remains almost the same, or at least do
not scale linearly with the network size.

\begin{figure}
  \centering
  \includegraphics[width=\columnwidth]{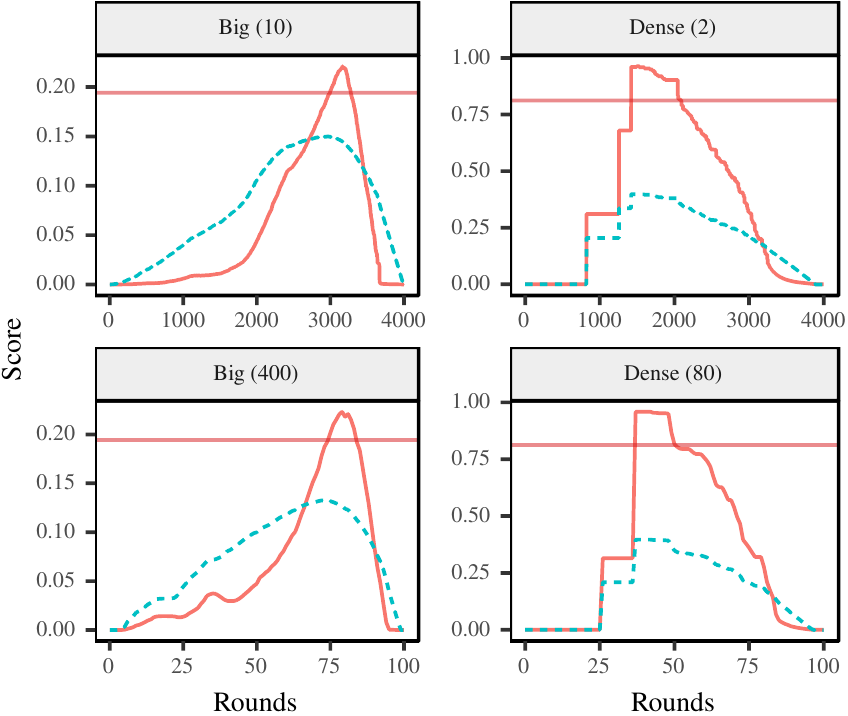}
  \caption{
    ARI (solid lines) and modularity (dashed line) scores along different
    rounds.  Input networks are exactly the same ones depicted in
    \cref{fig:bignet}.  Plots are arranged according to the amount of edges
    removed per round: $\approx 0.025\%$ and $1\%$.  The exact number of
    edges removed per round is indicated between parentheses.  For all runs, we
    set $\alpha = 0.1$ and the system runs until $t = 80$ for the dense network
    and $t = 250$ for the big network.  The solid horizontal line indicates the
    ARI score obtained by the CFG method.
  }
  \label{fig:bignet_nedges}
\end{figure}

Furthermore, we analyse the scores along different rounds by running up to the
iteration $t=80$ for the dense network and $t=250$ for the big network.  Results
are plotted in \cref{fig:bignet_nedges}.  Both the \mcf evolution and the ARI score
follow similar shapes compared with the results of smaller networks presented in
the previous sections.

Since the summation in Equation~\ref{eq:v} takes place on the set of
edges, the time complexity for a single dynamical iteration $t$ is $O(m)$.
Seeing that $t_\text{total}$ does not scale with the network size, the entire
run still takes $O(m)$. In addition, at the end of each round (a single
run or a small set of runs), we need to target the
highest-misaligned
edges, what requires sorting the set of edges and usually takes $O(m \,
\text{log} \, m)$. Since $n_\text{rounds}$ does not scale with the network size
either, then all the edge-removal process has a quasilinear time complexity:
$O(m \, \text{log} \, m)$ if $m \gg n$, or just $O(n \, \text{log} \, n)$ for
sparse networks.

However, it is important to emphasize that the cost $O(m \, \text{log} \, m)$
holds only from the asymptotic point of view, when networks become very large.
For small- or medium-sized networks, the cost caused by $t_\text{total}$ and
$n_\text{rounds}$ may become noticeable.

\subsection{Experiments on unbalanced-communities networks}

\begin{figure}
  \centering
  \includegraphics[width=\columnwidth]{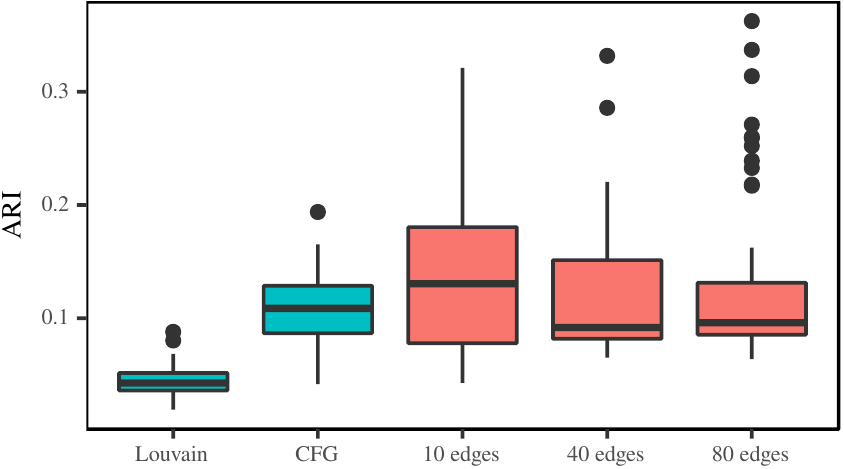}
  \caption{
    Results of removing 10, 40 and 80 edges per round on a population of 50
    different networks. Each network has 4 communities, containing 800, 400,
    200, and 100 vertices, $\adeg \approx 10$, and $\pin = 0.66$.  Boxes in blue
    are the results of applying CFG and Louvain on the network.  The remaining
    boxes (in red) are obtained by applying the proposed technique removing
    different amounts of edges.  The partition that yields the best modularity
    is chosen.  10 runs per round are performed, and parameter $\alpha = 0.1$
    for all simulations.  We stop the system at $t = 250$.
  }
  \label{fig:scores_asymmetric}
\end{figure}

All the experiments presented so far are performed on a class of balanced
networks, whose communities have the same size. In order to provide
more-realistic examples, let us introduce now some results on networks whose
communities have different number of vertices. These results are shown in
\cref{fig:scores_asymmetric}.

We notice now a considerable drop in the quality of the partitions. Still,
compared to traditional algorithms like CFG or Louvain applied to the original
network, the proposed edge-removal process performs well.  The performance of the
proposed method, however, varies greatly.  One possible explanation is that the
fixed number of iterations $t=250$ is not ideal.  The study of heuristics for
the system's stop condition are left for future works.


\subsection{Experiments on Lancichinetti benchmark}

Since real-world networks usually have heterogeneous degree distribution, we
also apply our technique on the benchmark of Lancichinetti \emph{et
al.}\cite{Lancichinetti:2008}.  Such benchmark produces networks in which both
degree and community size distributions follow power law functions with
arbitrary exponents.  Motivated by typical values found in natural systems, we
choose exponent $2$ for the degree distribution and $1$ for the size of
communities.  The generated networks have a mixing parameter $\mu$ that
controls the fraction of links between nodes of different communities.
To assess our method, we use networks with $1000$ nodes and average degree
$10$, varying the mixing parameter $\mu$ between $0.1$ and $0.6$.

We compare our method against CFG and Louvain in $30$ independent trials.  The
performance is measured in terms of the normalized mutual information index,
which measures the similarly of the predicted partition against the expected
one.  Moreover, such index is suggested by the authors of the benchmark.  CFG
and Louvain have no parameter, while in our technique we set $\alpha = 0.05$.
For stopping criterion, we run the system until the maximum change in each
projection $\hat{v}_{i,d}$ is less than $10^{-3}$ and remove the edge with
higher \mcf.  We interpret each connected component as a community.  We repeat
the edge-removal process until the modularity of the partitioning starts
decreasing.

\begin{figure}
  \centering
  \includegraphics{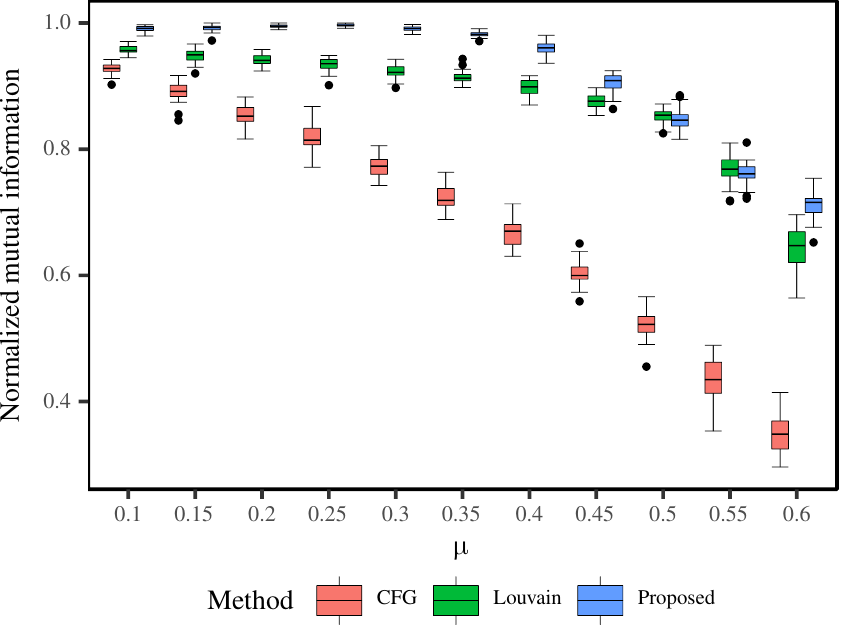}
  \caption{
    Comparison of the normalized mutual information index obtained by three
    community detection algorithms---our proposed method, CFG and Louvain---in
    the Lancichinetti \emph{et al.} benchmark\cite{Lancichinetti:2008}.  Results
    are taken from 30 independent trials for each value of the mixing parameter
    $\mu$.  Networks contain $1000$ nodes, average degree $10$.  Degree and
    community size distribution follow power law functions with exponents $2$
    and $1$, respectively.
  }
  \label{fig:lanc}
\end{figure}

\Cref{fig:lanc} reveals that our method performs significantly better
than CFG and Louvain.  Beyond the mark $\mu = 0.5$,
communities are no longer defined in a strong sense since each node has more
neighbors in other communities than in its own.

\subsection{Experiments on real-world datasets}

We apply our algorithm to four well-known real-world networks: Zachary's karate
club~\cite{Zachary1977}, the bottlenose dolphins social
network~\cite{Lusseau2003}, the American College Football~\cite{Girvan2002}, and
the Krebs’ books on US politics\footnote{This network is unpublished and can be
found at
\texttt{http://www-personal.umich.edu/{\raise.17ex\hbox{$\scriptstyle\sim$}}mejn/netdata/}}.

\begin{table}
  \centering
  \caption{
    Modularity score of community detection methods in 4 classic real-world
    datasets.  Optimal modularity is shown for comparison.
  }
  \begin{tabular}{c c c c c}
    \toprule
    Method & Dolphin & Football & Karate & Political Books \\
    \midrule
    Proposed                   & 0.529 & 0.605 & 0.419 & 0.527 \\
    Amiri~\cite{Amiri2013}     & 0.515 & 0.597 & 0.417 & 0.518 \\
    Honghao~\cite{Honghao2013} & 0.529 & 0.605 & 0.420 & 0.527 \\
    Song~\cite{Song2016}       &   -   & 0.531 & 0.362 & 0.463 \\
    \midrule
    Optimal & 0.529 & 0.605 & 0.420 & - \\
    \bottomrule
  \end{tabular}
  \label{tab:real}
\end{table}

\Cref{tab:real} presents the best modularity scores obtained by our method.  We
set $\alpha = 0.1$, $30$ independent runs, and vary the fraction of removed
edges in $1\%$, $2\%$, \dots, $10\%$.  Also, we vary the number of iterations $t
\in [30, 70]$.  We compare our results against three other bio-inspired
optimization methods.  (Missing results have not been measured by the authors of
the original paper.)  Also, the optimal value of modularity is calculated using
an exhaustive search in all possible partitions.  Modularity optimization is an
NP-complete problem, and known algorithms have exponential time
complexity~\cite{newman:2004}.  (We were not able to find out the optimal
partition of the \emph{political books} network.)

We observe that our method reached either optimal or nearly-optimal modularity
scores.  Although the ant-colony-based technique~\cite{Honghao2013} had similar
performance, our algorithm has lower computational cost because it explores the
solution space in a greedy manner.


\section{Conclusions}
\label{sec:conclusions}

Throughout an extensive set of experiments presented in this paper, we see that
the proposed flocking-like dynamical system and the iterative edge-removal
process performs well in many scenarios. The decentralized, self-organizing dynamical model is
robust, thus applicable on a wide variety of networks.

The concept of \mcf defined here is, in some sense, similar to the
concept of edge betweenness. High-misaligned edges are supposed to link distinct
communities. However, the cost $O(n^3)$ (on sparse networks) of the method
proposed by Newman and Girvan~\cite{newman:2004} can be prohibitive for its
application in large networks. On the other hand, we claim that our
edge-removal process is asymptotically quasilinear: $O(n \,
\text{log} \, n)$, which is quite attractive for large-network community
detection.

In further works, we will study heuristics to find out good values of the number
of iterations and removed edges per round.

\section*{Acknowledgments}

This research work was supported by the State of São Paulo Research Foundation
(FAPESP) (Projects 13/08666-8, 13/25876-6, and 15/50122-0), by the Brazilian
National Research Council (CNPq) (Project 303012/2015-3), and by the German
Research Foundation (DFG) (Project IRTG/GRK 1740).

R.A. Gueleri and F.A.N. Verri programmed and performed the computer experiments,
as well as produced the resultant plots and figures. F.A.N. Verri performed the
analytical study.  All authors contributed on conceiving the method, planning
the experiments, and writing this manuscript. L. Zhao, in addition, supervised all
this research work.

\bibliographystyle{IEEEtran}
\bibliography{IEEEabrv,references}

\end{document}